\documentclass[runningheads]{llncs}
\usepackage{amssymb}
\usepackage{graphicx}
\usepackage{tikz}
\usepackage{amsmath}
\usepackage{mathrsfs}
\usepackage{graphicx}
\usetikzlibrary{automata, positioning, arrows}
\usepackage{mathtools}
\usepackage{xcolor}
\usepackage{stackengine}

\begin{document}

\title{Generalized Circular One-Way Jumping Finite Automata}
\author{Ujjwal Kumar Mishra \and
Kalpana Mahalingam
\and
Rama Raghavan}
\authorrunning{ukmishra et al.}
\institute{Department of Mathematics, Indian Institute of Technology Madras,\\ Chennai-600036, India.\\
\email{ma16d030@smail.iitm.ac.in}\\
\email{\{kmahalingam,ramar\}@iitm.ac.in}}
\maketitle            
\begin{abstract}
A discontinuous model of computation called one-way jumping finite automata was defined by H. Chigahara et. al. This model was a restricted version of the model jumping finite automata. One-way jumping finite automata change their states after deleting a letter of an input and jump only in one direction. Allowing a state to delete a subword instead of a letter, we define a new model generalized circular one-way jumping finite automata. These automata work on an input in a circular manner. Similar to one-way jumping finite automata, generalized circular one-way jumping finite automata also jump only in one direction. We show that this newly defined model is powerful than one-way jumping finite automata. We define new variants(right and left) of the model generalized circular one-way jumping finite automata and compare them. We also compare the newly defined model with Chomsky hierarchy. Finally, we explore closure properties of the model.

\keywords{Jumping Finite Automata  \and One-Way Jumping Finite Automata \and Generalized Linear One-Way Jumping Finite Automata \and Generalized Circular One-Way Jumping Finite Automata.}

\end{abstract}

\section{Introduction}
General jumping finite automata $\it{(GJFA)}$, very first discontinuous model of computation, were introduced in \cite{JFA2012} by Meduna et. al. These automata read and delete an input in a discontinuous manner. The automata can jump in either direction to delete the input. The model was extensively studied in \cite{basicpropertiesJFA,trodip,oscadojfa,cacrojfa}. Following \cite{JFA2012}, several jumping transition models have been introduced and studied in the last few years \cite{owjfa,glowjfa,opvojfa,ndrowjfa,eoawjm,odjfaatcp,WKJFA,J53WKA,twja}. 

The model which is of our interest is one-way jumping finite automata ($\it{OWJFA}$), was defined in \cite{owjfa}. The model $\it{OWJFA}$ was defined by giving a restriction on jumping behaviour of the model jumping finite automata ($\it{JFA}$), a restricted variant of $\it{GJFA}$. There are two variants of $\it{OWJFA}$: right one-way jumping finite automata ($\it{ROWJFA}$) and left one-way jumping finite automata ($\it{LOWJFA}$). $\it{ROWJFA}$ starts processing an input string from the leftmost symbol of the input whereas $\it{LOWJFA}$ starts processing an input string from the rightmost symbol. The models can
only jump over symbols which they cannot process in its current state. Some properties of $\it{ROWJFA}$ are given in \cite{porowjfa}.  The decision problems of the model $\it{ROWJFA}$ are discussed in \cite{dorowjfa}. Nondeterministic variant of the model $\it{ROWJFA}$ is defined and studied in \cite{ndrowjfa}.

$\it{OWJFA}$ processes one symbol at a time, i.e., it go from one state to another by deleting a symbol. One extension of this model that changes its states by deleting a subword is generalized linear one-way jumping finite automata ($\it{GLOWJFA}$), was defined in \cite{glowjfa}. These automata work on a linear tape and can process a subword in a state. In this paper, we introduce a new model called generalized circular one-way jumping finite automata $(\it{GCOWJFA})$. This model is also a generalization, in deleting an input, of the model $\it{OWJFA}$. Unlike $\it{OWJFA}$, it can process a subword, i.e., $\it{GCOWJFA}$ changes its states by deleting a subword instead of a symbol. The difference between the model $\it{GLOWJFA}$ and $\it{GCOWJFA}$ is that $\it{GLOWJFA}$ work on a linear tape whereas $\it{GCOWJFA}$ work on a circular tape. Because of this difference both the models have distinct computational power.  We show that the generalization $\it{GCOWJFA}$ increases the power of $\it{OWJFA}.$ Similar to $\it{OWJFA}$, we define two types of $\it{GCOWJFA}$: generalized right circular one-way jumping finite automata ($\it{GRCOWJFA}$) and generalized left circular one-way jumping finite automata ($\it{GLCOWJFA}$). The models $\it{GRCOWJFA}$ and $\it{GLCOWJFA}$ delete an input in clockwise and anti-clockwise direction, respectively. Both the models $\it{GRCOWJFA}$ and $\it{GLCOWJFA}$ can jump over a part of the input if they cannot delete it.

In this paper, we compare the model $\it{GCOWJFA}$ with the models $\it{OWJFA}$ and $\it{GLOWJFA}$. We also compare $\it{GRCOWJFA}$ and $\it{GLCOWJFA}$. The language classes of $\it{GRCOWJFA}$ and $\it{GLCOWJFA}$ are compared with the language classes of Chomsky hierarchy. Closure properties of the language classes of $\it{GRCOWJFA}$ and $\it{GLCOWJFA}$ are explored.

The paper is organised as follows: In Section \ref{prelim}, we give some basic notion and notation. We also recall the definitions of $\it{ROWJFA}$, $\it{LOWJFA}$, $\it{GRLOWJFA}$, $\it{GLLOWJFA}$ and give examples of these models, in this section. The models $\it{GRCOWJFA}$ and $\it{GLCOWJFA}$ have been introduced in Section \ref{gcowjfa}. The models have been illustrated by examples. In Section \ref{gcgl}, the models $\it{GLOWJFA}$ and $\it{GCOWJFA}$ have been compared. We compare the models $\it{GRCOWJFA}$ and $\it{GLCOWJFA}$ in Section \ref{grcglc}. In Section \ref{gcch}, the language classes of the newly introduced model $\it{GCOWJFA}$ are compared with the language classes of Chomsky hierarchy. Finally, we discuss closure properties of the newly introduced models, in Section \ref{closure}. We end the paper with few concluding remarks.

\section{Preliminaries}\label{prelim} 
In this section, we recall some basic notations and definitions. An alphabet set is a finite non-empty set $\Sigma$. The elements of $\Sigma$ are called letters or symbols. A word or string $w=a_1a_2 \cdots a_n$ is a finite sequence of symbols, where $a_i \in \Sigma$ for $1 \leq i \leq n$. The reverse of $w$ is obtained by writing the symbols of $w$ in reverse order and denoted by $w^R$, hence $w^R=a_n \cdots a_2a_1$. By $\Sigma^*$, we denote the set of all words over the alphabet $\Sigma$ and by $\lambda$, the empty word. A language $L$ is a subset of $\Sigma^*$ and $L^\mathrm{C} = \Sigma^* \setminus L$ denotes the complement of $L$. The symbol $\emptyset$ represents the empty language or empty set. For an arbitrary word $w \in \Sigma^*$, we denote its length or the number of letters in it by $|w|$. For a letter $a \in \Sigma$, $| w |_a$ denotes the number of occurrences of $a$ in $w$. Note that, $\Sigma^+=\Sigma^* \setminus \{ \lambda \}$ and $|\lambda|=0$. For a finite set $A$, $|A|$ denotes the number of elements in $A$. A word $y \in \Sigma^*$ is a subword or substring of a word $w \in \Sigma^*$ if there exist words $x, z \in \Sigma^*$ such that $w=xyz$. If $w=uv$,
then the words $u \in \Sigma^*$ and $v \in \Sigma^*$ are said to be a prefix and a suffix of $w$, respectively. Two sets $A$ and $B$ are comparable if $A \subseteq B$ or $B \subseteq A$. For two sets $A$ and $B$, if $A$ is a proper subset of $B$, then we use the notation $A \subset B$. For definitions of other basic language operations (like union, intersection etc.), the reader is referred to \cite{book}.

We now recall the definitions of right and left one-way jumping finite automata \cite{owjfa}.
A right  one-way jumping finite automaton $(\it{ROWJFA})$ is a quintuple $\mathcal{A}=(\Sigma, Q, q_0, F, R)$, where $\Sigma$ is an alphabet set, $Q$ is a finite set of states, $q_0$ is the starting state, $F \subseteq Q$ is a set of final states and $R \subseteq Q \times \Sigma \times Q$ is a set of rules, where for a state $p \in Q$ and a symbol $a \in \Sigma$, there is at most one $q \in Q$ such that $(p,a,q) \in R$. By a rule $(p,a,q) \in R$, we mean that the automaton goes to the state $q$ from the state $p$ after deleting the symbol `$a$'.  For $p \in Q$, we set $$\Sigma_p=\Sigma_{R,p}=\{b \in \Sigma 
: (p,b,q) \in R ~ for~ some~ q \in Q\}.$$ A configuration of the right one-way jumping automaton $\mathcal{A}$ is a string of $Q\Sigma^*$. The right one-way jumping relation, denoted as $\circlearrowright_{\mathcal{A}}$, over $Q\Sigma^*$ is defined as follows. Let $(p,a,q) \in R, x \in (\Sigma \setminus \Sigma_p)^*$ and $y \in \Sigma^*$. Then, the $\it{ROWJFA}$ $\mathcal{A}$ makes a jump from the configuration $pxay$ to the configuration $qyx$, written as $pxay \circlearrowright _{\mathcal{A}} qyx$ or just $pxay \circlearrowright  qyx$ if it is clear which $\it{ROWJFA}$ is being referred. Let $\circlearrowright^+$ and $\circlearrowright^*$ denote the transitive and reflexive-transitive closure of $\circlearrowright$, respectively. The language accepted by the automaton $\mathcal{A}$ is $$L_R(\mathcal{A})=\{w \in \Sigma^* : q_0w \circlearrowright^* q_f ~ for~some ~ q_f \in F\}.$$

A left one-way jumping finite automaton $(\it{LOWJFA})$ is similar to that of a $\it{ROWJFA}$ except that a configuration of $\it{LOWJFA}$ is a string of $\Sigma^*Q$. By a rule $(q,a,p) \in R$, we mean that the automaton goes to the state $q$ from the state $p$ after deleting the symbol `$a$'. For $(q,a,p) \in R, x \in (\Sigma \setminus \Sigma_p)^*$ and $y \in \Sigma^*$ the  $\it{LOWJFA}$ makes a jump from the configuration $yaxp$ to the configuration $xyq$, written as $xyq \prescript{}{\mathcal{A}}{\circlearrowleft}$ $yaxp$ or just $xyq \circlearrowleft yaxp $ if it is clear which $\it{LOWJFA}$ is being referred.  The language accepted by the automaton $\mathcal{A}$ is $$L_L(\mathcal{A})=\{w \in \Sigma^* : q_f \prescript{*}{}{\circlearrowleft}~ wq_0 ~ for~some ~ q_f \in F\}.$$

\begin{example}
Consider the $\it{ROWJFA}$
$$
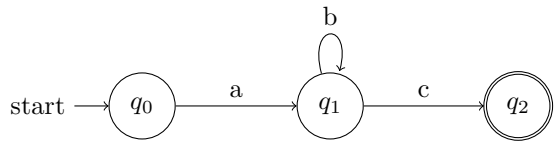
\begin{tikzpicture}
\node[state, initial] (q0) {$q_0$};
\node[state, right of=q0, xshift=1.5cm] (q1) {$q_1$};
\node[state, accepting, right of=q1, xshift=1.5cm] (q2) {$q_2$};
\path[->]
(q0) edge[above] node{a} (q1)
(q1) edge[loop above] node{b} (q1)
(q1) edge[above] node{c} (q2);
\end{tikzpicture}
$$

\noindent $\mathcal{A}=(\{a,b,c\},\{q_0,q_1,q_2\},q_0,\{q_2\},R)$, the set of rules $R$ is $R=\{(q_0,a, q_1),$ $(q_1,b, q_1),$ $(q_1,c,q_2)\}$. Here, $\Sigma_{q_0}=\{a\}$, $\Sigma_{q_1}=\{b,c\}$, $\Sigma_{q_2}=\emptyset$.
Since $\Sigma_{q_0}=\{a\}$,  the automaton will jump over any word from $(\Sigma\setminus\Sigma_{q_0})^*=\{b,c\}^*$ at $q_0$. Also, since $\Sigma_{q_1}=\{b,c\}$,  the automaton will jump over any word from $(\Sigma\setminus\Sigma_{q_1})^*=\{a\}^*$ at $q_1$. It is clear from the rules that if a word is in the language, then it will have exactly one $`a$' and one $`c$'. Moreover, the automaton has a loop labeled $`b$' at $q_1$. Therefore, the language will have the words of the form $b^lab^mcb^n$ or $b^lcb^mab^n$. We can divide these words in 8 groups $\{l=m=n=0\}$, $\{l=0, m=0, n \geq 1\}$, $\{l=0, m \geq 1, n=0\}$, $\{l=0, m \geq 1, n \geq 1 \}$, $\{l \geq 1, m=n=0\}$, $\{l \geq 1, m=0, n \geq 1\}$, $\{l \geq 1, m \geq 1, n=0\}$, $\{l \geq 1, m \geq 1, n \geq 1 \}$. Using case-by-case analysis, we conclude that the words $ab^mc$, where $m \geq 0$ and $b^lcab^n$, where $l,n \geq 0$ are in the language. For example, consider a string $b^lcab^n$, where $l,n \geq 0$. The sequence of transitions is 
$$q_0b^lcab^n \circlearrowright q_1b^nb^lc \circlearrowright^* q_1c  \circlearrowright q_2.$$

\noindent Hence, the language accepted by the automaton is 

$$L_R(\mathcal{A})=\{ab^mc~:~ m \geq 0\} \cup \{b^lcab^n~:~ l,n \geq 0\}.$$
\end{example}

\begin{note}
The language of the automaton, when it is considered as jumping finite automaton  (\cite{JFA2012}) is $\{b^lab^mcb^n~:~ l,m,n \geq 0\}$ $\cup$ $\{b^lcb^mab^n~:~ l,m,n \geq 0\}$.
\end{note}

A right one-way jumping finite automaton \cite{owjfa} can be  interpreted in two ways: one in which we think it working on linear tape and the other in which we think it working on a tape in a circular manner. When it is thought working on linear tape, in that case, the read head moves in one direction only and starts from the leftmost symbol of an input word. It moves from left to right (and possibly jumps over parts of the input) and upon reaching the end of the input word the automaton jumps to the leftmost symbol of the remaining concatenated word and starts deleting the word freshly from the last visited state. If a transition is defined for the current state and the next letter to be read, then the automaton deletes the letter. If not, but in the remaining input there are letters for which a transition is defined from the current state, the read head jumps to the nearest such letter to the right for its reading. The computation continues until all the letters are deleted or the automaton is stuck in a state in which it cannot delete any letter of the remaining word. 

In the light of this interpretation, an extension of $\it{ROWJFA}$ was given based on reading length of a word and generalized right linear one-way jumping finite automaton was defined in \cite{glowjfa}. A generalized right linear one-way jumping finite automaton($\it{GRLOWJFA}$) works on a linear tape and deletes an input word from left to right. The automaton starts deleting the input word with the reading head at the leftmost position of the input word. The automaton can jump over a part of the input word which it cannot process in its current state. It reads the nearest available subword of the word present on the input tape. If there is no transition available for current state, the automaton returns(head of the automaton returns) to the leftmost position of the current input and continues its computation. The computation continues until all the letters are deleted or the automaton is stuck in a state in which it cannot delete any subword of the remaining word. The generalized left linear one-way jumping finite automaton($\it{GLLOWJFA}$) work similar to $\it{GRLOWJFA}$. The read head of a $\it{GLLOWJFA}$ moves from right to left. We recall the formal definitions of generalized right linear one-way jumping finite automaton and generalized left linear one-way jumping finite automaton from \cite{glowjfa}, given below.

\begin{definition}\label{grld}
A generalized right linear one-way jumping finite automaton$~~$ ($\it{GRLOWJFA}$) is a tuple $\mathcal{A}=(\Sigma,Q,q_0,F, R)$, where $\Sigma,Q,q_0,F$ are same as $\it{ROWJFA}$ and $R \subset Q \times \Sigma^+ \times Q$ is a finite set of rules, where for a state $p \in Q$ and a word $w \in \Sigma^+$, there is at most one $q \in Q$ such that $(p,w,q) \in R$.  By a rule $(p,w,q) \in R$, we mean that the automaton goes to the state $q$ from the state $p$ after deleting the word `$w$'. The automaton is deterministic in the sense that for a state $p$ and for a word $w \in \Sigma^+$ we have at most one $q \in Q$ such that $(p,w,q) \in R$. A configuration of the automaton $\mathcal{A}$ is a string of $\Sigma^* Q\Sigma^*$. The generalized right linear one-way jumping relation, denoted as $\curvearrowright_{\mathcal{A}}$, or just $\curvearrowright$ if it is clear which $\it{GRLOWJFA}$ is being referred, over $\Sigma^*Q\Sigma^*$ is defined as follows: For a state $p \in Q$, set $\Sigma_p=\Sigma_{R,p}=\{w \in \Sigma^+ 
: (p,w,q) \in R$ for some $q \in Q\}$.

\begin{itemize}
    \item[1.] Let $t,u,v \in \Sigma^*$ and $(p,x,q) \in R$. Then the $\it{GRLOWJFA}$ $\mathcal{A}$ makes a jump from the configuration $tpuxv$ to the configuration $tuqv$, written as $$tpuxv \curvearrowright tuqv$$ 
if $u$ does not contain any word from $\Sigma_p$ as a subword, i.e., $u \neq u'wu''$,
where $u',u'' \in \Sigma^*,~w \in \Sigma_p$ and $u_2x_1 \neq x$, where $u_2, x_1 \in \Sigma^+$ and $u=u_1u_2,~x=x_1x_2$ for some $u_1,x_2 \in \Sigma^*$. This is what we mean by nearest available subword.
\item[2.] Let $x \in \Sigma^+,y \in \Sigma^*$ and $y$ does not contain any word from $\Sigma_p$ as a subword, i.e., $y \neq y_1wy_2$, where $y_1,y_2 \in \Sigma^*$ and $w \in \Sigma_p$, then the $\it{GRLOWJFA}$ $\mathcal{A}$ makes a jump from the configuration $xpy$ to the configuration $pxy$, written as $$xpy \curvearrowright pxy.$$ 
\end{itemize} 

The language accepted by the $\it{GRLOWJFA}$ $\mathcal{A}$ is $$L_{GRL}(\mathcal{A})=\{w \in \Sigma^* : q_0w \curvearrowright^* q_f ~ for~some ~ q_f \in F\}.$$
\end{definition}
Similar to that of $\it{GRLOWJFA}$, the generalized left linear one-way jumping finite automaton was defined.

\begin{definition}\label{glld}
A generalized left linear one-way jumping finite automaton denoted by $\it{GLLOWJFA}$ is a tuple $\mathcal{A}=(\Sigma,Q,q_0,F, R)$, where $\Sigma,Q,q_0,F$ are same as $\it{ROWJFA}$ and $R \subset Q \times \Sigma^+ \times Q$ is a finite set of rules, where for a state $p \in Q$ and a word $w \in \Sigma^+$, there is at most one $q \in Q$ such that $(q,w,p) \in R$. By a rule $(q,w,p) \in R$, we mean that the automaton goes to the state $q$ from the state $p$ after deleting the word `$w$'. The automaton is deterministic in the sense that for a state $p$ and for a word $w \in \Sigma^+$ we have at most one $q \in Q$ such that $(q,w,p) \in R$. A configuration of the automaton $\mathcal{A}$ is a string of $\Sigma^*Q\Sigma^*$. The generalized left linear one-way jumping relation, denoted as $\prescript{}{{\mathcal{A}}}{\curvearrowleft}$, or just $\curvearrowleft$ if it is clear which $\it{GLLOWJFA}$ is being referred, over $\Sigma^*Q\Sigma^*$ is defined as follows: For a state $p \in Q$, set $\Sigma_p=\Sigma_{R,p}=\{w \in \Sigma^+ 
: (q,w,p) \in R$ for some $q \in Q\}$.
\begin{itemize}
    \item[1.] Let $t,u,v \in \Sigma^* $ and $(q,x,p) \in R$. Then, the $\it{GLLOWJFA}$ $\mathcal{A}$ makes a jump from the configuration $vxupt$ to the configuration $vqut$, written as $$vqut \curvearrowleft vxupt$$ 
if $u$ does not contain any word from $\Sigma_p$ as a subword, i.e., $u \neq u'wu''$,
where $u',u'' \in \Sigma^*,~w \in \Sigma_p$ and $x_2u_1 \neq x$, where $u_1, x_2 \in \Sigma^+$ and $u=u_1u_2,~x=x_1x_2$ for some $u_2,x_1 \in \Sigma^*$. 

\item[2.] Let $x \in \Sigma^+,y \in \Sigma^*$ and $y$ does not contain any word from $\Sigma_p$ as a subword, i.e., $y \neq y_1wy_2$, where $y_1,y_2 \in \Sigma^*$ and $w \in \Sigma_p$, then the $\it{GLLOWJFA}$ $\mathcal{A}$ makes a jump from the configuration $ypx$ to the configuration $yxp$, written as $$yxp \curvearrowleft ypx.$$
\end{itemize}
The language accepted by the $\it{GLLOWJFA}$ $\mathcal{A}$ is $$L_{GLL}(\mathcal{A})=\{w \in \Sigma^* : q_f~ \prescript{*}{}{\curvearrowleft}~ wq_0 ~ for~some ~ q_f \in F\}.$$

\end{definition}

We illustrate Definitions \ref{grld} and \ref{glld} with the following example.

\begin{example}\label{exgl}
Consider the automaton $\mathcal{A}=(\{a,b,c\},\{q_0,q_1\},q_0,\{q_1\},R)$, where $R$ is depicted in the figure below.

$$
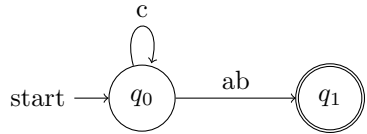
\begin{tikzpicture}
\node[state, initial] (q0) {$q_0$};
\node[state, accepting, right of=q0, xshift=1.5cm] (q1) {$q_1$};
\path[->] (q0) edge[loop above] node{c} (q0)
(q0) edge  [above] node{ab} (q1);
\end{tikzpicture}
$$
We first consider the automaton as a $\it{GRLOWJFA}$. Then, the set of rules is $R=\{(q_0,c,q_0), (q_0,ab,q_1) \}$. Hence, $\Sigma_{q_0}=\{c,ab\}$ and $\Sigma_{q_1}= \emptyset$.
 Since $\Sigma_{q_0}=\{c,ab\}$, the automaton at $q_0$ can neither jump over a `$c$' nor over `$ab$'. The automaton can delete arbitrary number of $c$'s at state $q_0$ and can jump $`a$' or $`b$' at $q_0$.  
Consider a word $c^lac^mbc^n$, where $l,n \geq 0, m \geq 1$. Then, $q_0c^lac^mbc^n \curvearrowright^* q_0ac^mbc^n \curvearrowright aq_0c^{m-1}bc^n \curvearrowright^*
 aq_0bc^n \curvearrowright abq_0c^{n-1} \curvearrowright^* abq_0 \curvearrowright q_0ab \curvearrowright q_1$
and hence,
$$L_{GRL}(\mathcal{A})=\{c^lab~|~ n \geq 0 \} \cup \{c^lac^mbc^n~|~ l,n \geq 0, m \geq 1\}.$$
Now, consider the automaton as $\it{GLLOWJA}$. In this case, the set of rules is $R=\{(q_0,c,q_0),(q_1,ab,q_0)\}$. It can be shown, as explained above, the language accepted by the automaton $\it{GLLOWJFA}$ is  
 $$L_{GLL}(\mathcal{A})=\{abc^n~|~ n \geq 0 \} \cup \{c^lac^mbc^n~|~ l,n \geq 0, m \geq 1\}.$$

\end{example}

\section{Generalized Circular One-Way Jumping Finite Automata}\label{gcowjfa}

The right one-way jumping automata \cite{owjfa} can also be thought/interpreted as working on a tape in a circular manner. In this case, a right one-way jumping finite automaton works on a circular tape and the read head moves in the clockwise direction. The read head starts from the first symbol of an input, i.e., read head is at $w_1$ of the input $w=w_1w_2 \cdots w_n$. The automaton can jump over parts of the input. If a transition is defined for a current state and the next letter to be read, then the automaton deletes the symbol. If not and in the remaining input there are letters for which a transition is defined from the current state, the read head jumps to the nearest such letter to delete it. The computation continues until all the letters are deleted or the automaton is stuck in a state in which it cannot delete any letter of the remaining word.

Keeping this interpretation in mind, we extend the definition of $\it{ROWJFA}$ based on reading length of a word and define generalized right circular one-way jumping finite automata. A generalized right circular one-way jumping finite automaton works on a circular tape and deletes an input word in the clockwise direction. The automaton starts deleting the input word with the reading head at the first symbol, the symbol at which reading head is present, of the input, ie., $w_1$ of $w=w_1w_2 \cdots w_n$. The automaton can jump over a part of the input word that it cannot delete. After deleting a subword of the word present on the input tape, the word is circularly permuted and the reading head goes to the first symbol just after the deleted subword. Suppose $uxv$ is present on the tape and the automaton deletes $x$, then the word is circularly permuted to get $vu$ and the reading head goes to the first symbol of $vu$.   The computation continues until all the letters are deleted or the automaton is stuck in a state in which it cannot delete any subword of the remaining word present on the tape. The $\it{GLCOWJFA}$ works similar to that of $\it{GRCOWJFA}$. The tape head of a $\it{GLCOWJFA}$ moves in the anti-clockwise direction. The formal definitions of generalized right circular one-way jumping finite automaton and generalized left circular one-way jumping finite automaton are given below.

\begin{definition}\label{grcd}
A generalized right circular one-way jumping finite automaton$~~$ ($\it{GRCOWJFA}$) is a tuple $\mathcal{A}=(\Sigma,Q,q_0,F, R)$, where $\Sigma,Q,q_0,F$ are same as $\it{ROWJFA}$ and $R \subset Q \times \Sigma^+ \times Q$ is a finite set of rules, where for a state $p \in Q$ and a word $w \in \Sigma^+$, there is at most one $q \in Q$ such that $(p,w,q) \in R$.  By a rule $(p,w,q) \in R$, we mean that the automaton goes to the state $q$ from the state $p$ after deleting the word `$w$'. The automaton is deterministic in the sense that for a state $p$ and for a word $w \in \Sigma^+$ we have at most one $q \in Q$ such that $(p,w,q) \in R$.
A configuration of the automaton $\mathcal{A}$ is a string of $Q\Sigma^*$. The generalized right circular one-way jumping relation, denoted as $\circlearrowright_{\mathcal{A}}$, or just $\circlearrowright$ if it is clear which $\it{GRCOWJFA}$ is being referred, over $Q\Sigma^*$ is defined as follows: For a state $p \in Q$, set $\Sigma_p=\Sigma_{R,p}=\{w \in \Sigma^+ 
: (p,w,q) \in R$ for some $q \in Q\}$. Let $u,v \in \Sigma^*$ and $(p,x,q) \in R$. Then, the $\it{GRCOWJFA}$ $\mathcal{A}$ makes a jump from the configuration $puxv$ to the configuration $qvu$, written as $$puxv \circlearrowright qvu$$ 
if $u$ does not contain any word from $\Sigma_p$ as a subword, i.e., $u \neq u'wu''$,
where $u',u'' \in \Sigma^*,~w \in \Sigma_p$ and $u_2x_1 \neq x$, where $u_2, x_1 \in \Sigma^+$ and $u=u_1u_2,~x=x_1x_2$ for some $u_1,x_2 \in \Sigma^*$. The language accepted by the $\it{GRCOWJFA}$ $\mathcal{A}$ is $$L_{GRC}(\mathcal{A})=\{w \in \Sigma^* : q_0w \circlearrowright^* q_f ~ for~some ~ q_f \in F\}.$$
\end{definition}
Similar to that of $\it{GRCOWJFA}$, we define the notion of a generalized left circular one-way jumping finite automaton as below.

\begin{definition}\label{glcd}
A generalized left circular one-way jumping finite automaton denoted by $\it{GLCOWJFA}$ is a tuple $\mathcal{A}=(\Sigma,Q,q_0,F, R)$, where $\Sigma,Q,q_0,F$ are same as $\it{ROWJFA}$ and $R \subset Q \times \Sigma^+ \times Q$ is a finite set of rules, where for a state $p \in Q$ and a word $w \in \Sigma^+$, there is at most one $q \in Q$ such that $(q,w,p) \in R$. By a rule $(q,w,p) \in R$, we mean that the automaton goes to the state $q$ from the state $p$ after deleting the word `$w$'. The automaton is deterministic in the sense that for a state $p$ and for a word $w \in \Sigma^+$ we have at most one $q \in Q$ such that $(q,w,p) \in R$. A configuration of the automaton $\mathcal{A}$ is a string of $\Sigma^*Q$. The generalized left circular one-way jumping relation, denoted as $\prescript{}{{\mathcal{A}}}{\circlearrowleft}$, or just $\circlearrowleft$ if it is clear which $\it{GLCOWJFA}$ is being referred, over $\Sigma^*Q$ is defined as follows: For a state $p \in Q$, set $\Sigma_p=\Sigma_{R,p}=\{w \in \Sigma^+ 
: (q,w,p) \in R$ for some $q \in Q\}$. Let $u,v \in \Sigma^* $ and $(q,x,p) \in R$. Then, the $\it{GLCOWJFA}$ $\mathcal{A}$ makes a jump from the configuration $vxup$ to the configuration $uvq$, written as $$uvq \circlearrowleft vxup$$ 
if $u$ does not contain any word from $\Sigma_p$ as a subword, i.e., $u \neq u'wu''$,
where $u',u'' \in \Sigma^*,~w \in \Sigma_p$ and $x_2u_1 \neq x$, where $u_1, x_2 \in \Sigma^+$ and $u=u_1u_2,~x=x_1x_2$ for some $u_2,x_1 \in \Sigma^*$. The language accepted by the $\it{GLCOWJFA}$ $\mathcal{A}$ is $$L_{GLC}(\mathcal{A})=\{w \in \Sigma^* : q_f~ \prescript{*}{}{\circlearrowleft}~ wq_0 ~ for~some ~ q_f \in F\}.$$
\end{definition}

We give examples illustrating Definitions \ref{grcd} and \ref{glcd}.

\begin{example}
Consider the automaton from Example \ref{exgl}.
$$
\begin{tikzpicture}
\node[state, initial] (q0) {$q_0$};
\node[state, accepting, right of=q0, xshift=1.5cm] (q1) {$q_1$};
\path[->] (q0) edge[loop above] node{c} (q0)
(q0) edge  [above] node{ab} (q1);
\end{tikzpicture}
$$
We first consider the automaton as a $\it{GRCOWJFA}$. Then, the set of rules is $R=\{(q_0,c,q_0),(q_0,ab,q_1)\}.$ Hence, $\Sigma_{q_0}=\{c,ab\}$ and $\Sigma_{q_1}= \emptyset$.
 Since $\Sigma_{q_0}=\{c,ab\}$, the automaton at $q_0$ can neither jump over a `$c$' nor over `$ab$'. The automaton can delete arbitrary number of $c$'s at state $q_0$ and can jump over $`a$' or $`b$' at $q_0$. So any word in the language will be of the form $c^lac^mbc^n$ or $c^lbc^mac^n$. We can divide these words in 8 groups $\{l=m=n=0\}$, $\{l=0, m=0, n \geq 1\}$, $\{l=0, m \geq 1, n=0\}$, $\{l=0, m \geq 1, n \geq 1 \}$, $\{l \geq 1, m=n=0\}$, $\{l \geq 1, m=0, n \geq 1\}$, $\{l \geq 1, m \geq 1, n=0\}$, $\{l \geq 1, m \geq 1, n \geq 1 \}$. And using case-by-case analysis, we conclude that the language accepted by the automaton is

\vspace{0.2cm}

\noindent $L_{GRC}(\mathcal{A})=\{c^lab~|~ l \geq 0 \} \cup \{c^lac^mbc^n~|~ l \geq 0, m,n \geq 1\} \cup \{c^lbc^ma~|~ l \geq 0, m \geq 1\}.$

Now, consider the automaton as $\it{GLCOWJA}$. In this case, the set of rules is $R=\{(q_0,c,q_0),(q_1,ab,q_0)\}$. It can be shown, using the above technique, the language accepted by the automaton $\it{GLCOWJFA}$ is 

\vspace{0.2cm}
\noindent $L_{GLC}(\mathcal{A})=\{abc^n~|~ n \geq 0 \} \cup \{c^lac^mbc^n~|~ l,m \geq 1, n \geq 0\} \cup \{bc^mac^n~|~ m \geq 1, n \geq 0\}.$

\end{example}

\begin{note}
Note that the language accepted by the automaton of Example \ref{exgl}, when it is considered as $\it{GJFA}$ (\cite{JFA2012}), is $L=\{c^lac^mbc^n~:~ l,m,n \geq 0\}$. The language accepted differs in $\it{GJFA}$, $\it{GRLOWJFA}$, $\it{GRCOWJFA}$, $\it{GLLOWJFA}$, $\it{GLCOWJFA}$ because of different ways of processing the same input.
\end{note}
\noindent Now, we give and example that shows that the language classes $\bf{GRCOWJ}$ and $\bf{GLCOWJ}$ are not disjoint, where $\bf{GRCOWJ}$ and $\bf{GLCOWJ}$ are the language classes accepted by $\it{GRCOWJFA}$ and $\it{GLCOWJFA}$ respectively.
\begin{example}
Consider the automaton
$$
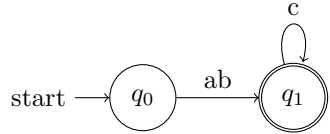
\begin{tikzpicture}
\node[state, initial] (q0) {$q_0$};
\node[state, accepting, right of=q0, xshift=1cm] (q1) {$q_1$};
\path[->] (q0) edge[above] node{ab} (q1)
(q1) edge[loop above] node{c} (q1);
\end{tikzpicture}
$$
The automaton cannot make any move without deleting $ab$ and at the state $q_1$ it can delete any number of $c$'s. Hence, the language accepted by the automaton is 
$$L_{GRC}(\mathcal{A})=L_{GLC}(\mathcal{A})=L_{GRL}(\mathcal{A})=L_{GLL}(\mathcal{A})=c^*abc^*.$$
\end{example}
Hence, we conclude the following.

\begin{note}
$\bf{GRCOWJ} \cap \bf{GLCOWJ}  \cap \bf{GRLOWJ} \cap \bf{GLLOWJ} \neq \emptyset$.
\end{note}

From the Definition \ref{grcd} of $\it{GRCOWJFA}$, we observe the following:
\begin{lemma}\label{confi}
Let $\mathcal{A}=(\Sigma,Q,q_0,F,R)$ be a $\it{GRCOWJFA}$. Let $w_1,w_2 \in \Sigma^*$ and $(p,x,q) \in R$. If the automaton goes from a configuration $pw_1xw_2$ to a configuration $qw_2w_1$, then the automaton also goes to the configuration $qw_2w_1$ from the configuration $pxw_2w_1$, i.e., if $pw_1xw_2$ $\circlearrowright$ $qw_2w_1$, then $pxw_2w_1$ $\circlearrowright$ $qw_2w_1$.
\end{lemma}
A similar observation can also be made for the model $\it{GLCOWJFA}$.

\section{GCOWJFA and GLOWJFA}\label{gcgl}

In this section, we compare the language class $\bf{GRCOWJ}$ with the language classes $\bf{ROWJ}$ and $\bf{GRLOWJ}$, where $\bf{ROWJ}$ and $\bf{GRLOWJ}$ are the language classes accepted by $\it{ROWJFA}$ and $\it{GRLOWJFA}$ respectively. The comparison between language classes $\bf{ROWJ}$ and $\bf{GRLOWJ}$ was done in \cite{glowjfa} and it was proved that $\bf{ROWJ \subset GRLOWJ}$ \cite{glowjfa}. Here, we show that the language class $\bf{ROWJ}$ is properly contained in the language class $\bf{GRCOWJ}$. Similarly, it can be shown that $\bf{LOWJ} \subset$ $\bf{GLCOWJ}$. We also show that $\bf{GRLOWJ}$ is not a part of $\bf{GRCOWJ}$. 

\noindent It is clear from the definitions of $\it{ROWJFA}$ and $\it{GRCOWJFA}$/$\it{GRLOWJFA}$ that when the rules of a $\it{GRCOWJFA}$/$\it{GRLOWJFA}$ satisfies the condition: if $(p,w,q) \in R$, then $|w|=1$. Then, the $\it{GRCOWJFA}$/$\it{GRLOWJFA}$ is nothing but $\it{ROWJFA}$. Hence, $\bf{ROWJ}$ $\subseteq$  $\bf{GRCOWJ} \cap \bf{GRLOWJ}$. Similarly, $\bf{LOWJ}$ $\subseteq$  $\bf{GLCOWJ}$ $\cap$ $\bf{GLLOWJ}$. Now, we give a language that is in the class $\bf{GRCOWJ}$ but not in the class $\bf{ROWJ}$.

\begin{example}\label{nonrowj}
(\cite{glowjfa}) Consider the $\it{GRCOWJFA}$  $\mathcal{A}=(\{a,b\}, \{q_0,q_1,q_2,q_3,q_4\},q_0,$  $\{q_0,q_1,$ $q_2,q_4\},$ $R)$, where the set of rules is $R=\{(q_0,a,q_1),(q_1,b,q_2),(q_2,a,q_3),$ $(q_3,b,q_2),$ $(q_0,aa,q_4),$ $(q_4,a,q_4)\}$. 
The language accepted by the automaton is

$L_{GRC}(\mathcal{A})=\{w \in \{a,b\}^*~|~|w|_a=|w|_b$ or $|w|_b=0\}$.
\end{example}
It was proved in \cite{porowjfa} that the language of Example \ref{nonrowj} cannot be accepted by any $\it{ROWJFA}$. It was shown in \cite{glowjfa} that the language is in the class $\bf{GRLOWJ}$. Hence, by Example \ref{nonrowj} and definitions of $\it{GRCOWJFA}$ and $\it{GRLOWJFA}$, we have the following result.

\begin{note}\label{owjgcowj}
$\bf{ROWJ} \subset \bf{GRCOWJ} \cap \bf{GRLOWJ}$.
\end{note}
Similarly, it can be shown that $\bf{LOWJ} \subset \bf{GLCOWJ} \cap \bf{GLLOWJ}$.

Now, we give a language that is in the language class $\bf{GRLOWJ}$ but not in the class $\bf{GRCOWJ}$.
\begin{example}\label{GRLD}
(\cite{glowjfa}) Consider the $\it{GRLOWJFA}$ automaton $\mathcal{A}=(\{a,b\},\{q_0\},q_0,$ $\{q_0\},$ $\{(q_0,ab,q_0)\})$.
The language accepted by the automaton is the Dyck language, $D$. Hence, $D \in$ $\bf{GRLOWJ}$.
\end{example}
Now, we prove that the Dyck language cannot be accepted by any $\it{GRCOWJFA}$.

\begin{lemma}\label{GRCD}
There does not exist any $\it{GRCOWJFA}$ that accepts the Dyck language, i.e., $D \notin \bf{GRCOWJ}.$
\end{lemma}

\begin{proof}
If $D \in \bf{GRCOWJ}$, then there exists a $\it{GRCOWJFA}$ $\mathcal{A}=(\Sigma,Q,$ $q_0,$ $F,$ $R)$ such that $L(\mathcal{A})=D$. Let $l=|Q|$ and $m=max\{|w|~:~ (p,w,q) \in R\}$. Choose $n > lm$. We know $a^nb^n \in D=L(\mathcal{A})$. Then, there exists a sequence of transitions such that $q_0a^nb^n \circlearrowright^* q_f$, where $q_f \in F$. In order to reach the final configuration $q_f$ from the initial configuration $q_0a^nb^n$, there will be intermediate configurations of the form $q'w_1a$ or $q''w_2b$ or both, where $q',q'' \in Q, w_1,w_2 \in \Sigma^*$.

\begin{itemize}
    \item[Case 1.] If there is an intermediate configuration of the form $q'w_1a$. Then, the automaton will make the following sequence of transitions $q_0a^nb^n \circlearrowright^* q'w_1a \circlearrowright^* q_f$. Then, by repeated application of Lemma \ref{confi}, we have a word $w \in \Sigma^*$ such that $q_0ww_1a \circlearrowright^* q'w_1a \circlearrowright^* q_f$. Hence, $ww_1a \in L(\mathcal{A})$ but $ww_1a \not\in D$ because no word of $D$ ends with `$a$'. Hence, there cannot be any intermediate configuration of the form $q'w_1a$. 
    
    \item[Case 2.] Since there cannot be any intermediate configuration of the form $q'w_1a$, all intermediate configurations have the form $q''w_2b$. In this case, the automaton will delete all $a$'s first and then $b$'s, starting from the initial configuration $q_0a^nb^n$, otherwise the automaton will go to a configuration of the form $q'w_1a$ because of Lemma \ref{confi}. Because of the choice of $n$, there will be a loop consisting of only $a$'s to delete all $a$'s without deleting $b$'s, starting from the initial configuration. The same loop can be used multiple times to create a mismatch in the indices of $`a$' and $`b$'. Hence, there will be a word of the form $a^{n_1}b^{n}$, $n_1 \neq n$, in $L(\mathcal{A})$ but $a^{n_1}b^{n} \not\in D$ for $n_1 \neq n$.
\end{itemize}
From the above discussed cases, we conclude that $D \not\in \bf{GRCOWJ}$.
\end{proof}
We conclude the following from Example \ref{GRLD} and Lemma \ref{GRCD}.
\begin{proposition}
The language class $\bf{GRLOWJ}$ is not contained in the language class $\bf{GRCOWJ}$, i.e., $\bf{GRLOWJ}$ $\not\subseteq$ $\bf{GRCOWJ}$. Similarly, it can be proved that $\bf{GLLOWJ}$ $\not\subseteq$ $\bf{GLCOWJ}$. 
\end{proposition}

\section{GRCOWJFA and GLCOWJFA}\label{grcglc}

In this section, the comparison between the language classes $\bf{GRCOWJ}$ and $\bf{GLCOWJ}$ has been done. We begin the section with a relationship that connects the languages of the class $\bf{GRCOWJ}$ and the languages of the class $\bf{GLCOWJ}$.

First we recall a similar result from \cite{glowjfa}.
\begin{lemma}
For a language $L_1 \in \bf{GRLOWJ}$, there exists a language $L_2 \in \bf{GLLOWJ}$ such that $L_1^R=L_2$ and vice versa.
\end{lemma}

Now, we show that for every language $L_1 \in \bf{GRCOWJ}$ $(\bf{GLCOWJ})$, there exists a language $L_2 \in \bf{GLCOWJ}$ $(\bf{GRCOWJ}$ resp) such that $L_1^R=L_2$.

\begin{proposition}
For a language $L_1 \in \bf{GRCOWJ}$, there exists a language $L_2 \in \bf{GLCOWJ}$ such that $L_1^R=L_2$ and vice versa.
\end{proposition}
\begin{proof}
Let $L_1 \in \bf{GRCOWJ}$. Then, there exists a $\it{GRCOWJFA}$, say $\mathcal{A}_1=(\Sigma,Q,q_0,F,R)$, such that $L_1=L(\mathcal{A}_1)$. We construct a $\it{GLCOWJFA}$, say $\mathcal{A}_2$, as: $\mathcal{A}_2=(\Sigma,Q,q_0,F,R')$, where $R'=\{(q,w^R,p)~|~ (p,w,q) \in R\}$. Let $L(\mathcal{A}_2)=L_2$. Claim, $L_1^R=L_2$. We have the following lemma:

\begin{lemma}\label{xpyux}
If $x,y \in \Sigma^*$, $u \in \Sigma^+$ and $p,q \in Q$, then $pxuy \circlearrowright_{\mathcal{A}_1} qyx$ if and only if $x^Ry^Rq$ $\prescript{}{{\mathcal{A}_2}}{\circlearrowleft}$  $y^Ru^Rx^Rp$.
\end{lemma}
\begin{proof}
Let $x,y \in \Sigma^*$, $u \in \Sigma^+$ and $p,q \in Q$. Now, $pxuy \circlearrowright_{\mathcal{A}_1} qyx$ if and only if $(p,u,q) \in R$, $x$ does not contain any word from $\Sigma_{R,p}$ as a subword and $x''u' \neq u$, where $x''$ is a nonempty suffix of $x$ and $u'$ is a nonempty prefix of $u$ if and only if $(q,u^R,p) \in R'$, $x^R$ does not contain any word from $\Sigma_{R',p}$ as a subword and $u^{(2)}x^{(1)} \neq u^R$, where $x^{(1)}$ is a nonempty prefix of $x^R$ and $u^{(2)}$ is a nonempty suffix of $u^R$ if and only if $x^Ry^Rq$ $\prescript{}{{\mathcal{A}_2}}{\circlearrowleft}$ $y^Ru^Rx^Rp$. Hence, for $x,y \in \Sigma^*$, $u \in \Sigma^+$ and $p,q \in Q$, $pxuy \circlearrowright_{\mathcal{A}_1} qyx$ if and only if $x^Ry^Rq$ $\prescript{}{{\mathcal{A}_2}}{\circlearrowleft}$   $y^Ru^Rx^Rp$. 
\end{proof}

Now, from Lemma \ref{xpyux}, $w \in L_1^R$ if and only if $q_0w^R \circlearrowright_{\mathcal{A}_1}^* q_f$, where $q_f \in F$, if and only if $q_f$ $\prescript{*}{{\mathcal{A}_2}}{\circlearrowleft}$ $wq_0$ if and only if $w \in L_2$. Hence, $L_1^R=L_2$.
\end{proof}

\begin{corollary}
If a language $L \in \bf{GRCOWJ \cap GLCOWJ}$, then the reversal of the language $L^R$ $\in$ $\bf{GRCOWJ \cap GLCOWJ}$.
\end{corollary}

Now, we compare the language class $\bf{GRCOWJ}$ with the language class $\bf{GLCOWJ}$. We give a language that is in the language class $\bf{GLCOWJ}$. We show that the language is not in the class $\bf{GRCOWJ}$. Hence, we conclude $\bf{GLCOWJ}$ is not contained in $\bf{GRCOWJ}$. Similarly, it can be shown that $\bf{GRCOWJ}$ is not contained in $\bf{GLCOWJ}$.

\begin{example}\label{Lc}
Consider the $\it{GLCOWJFA}$ $\mathcal{A}=(\{a,b,c\},\{q_0,q_1,q_2\},q_0,$ $\{q_1\},$ $R)$, where $R=\{(q_2,a,q_0), (q_2,b,q_0), (q_1,c,q_0),(q_1,ba,q_1)\}$. Here, $\Sigma_{q_0}=\{a,b,c\}$, $\Sigma_{q_1}=\{ba\}$,$\Sigma_{q_2}=\emptyset$.
We have the following observations regarding the language $L(\mathcal{A})$ of the automaton. Let $L(\mathcal{A})=L_c$.
\begin{itemize}
\item[$\bullet$] Note that all words of $L(\mathcal{A})$ will end with $c$. Suppose there is a word $w \in \Sigma^*$ that ends with $a$. Then, $w=w_1a$, where $w_1 \in \Sigma^*$. Since $a \in \Sigma_{q_0}$, the automaton cannot jump over `$a$' at $q_0$ and hence, after deleting `$a$' the automaton will go to the state $q_2$, i.e., $w_1q_2$ $\circlearrowleft$ $w_1aq_0$. Hence, $w=w_1a \not\in L(\mathcal{A})$. Same argument holds true for words ending with $`b$'.

\item[$\bullet$] Note that, $a^nb^na^nb^nc \in L(\mathcal{A})$ for $n \geq 0$ since, $q_1$ $\circlearrowleft$ $\cdots$ $\circlearrowleft$ $a^{n-1}b^{n-1}a^{n-1}b^{n-1}q_1$   $\circlearrowleft$ $a^{n-1}b^na^nb^{n-1}q_1$   $\circlearrowleft$ $a^nb^na^nb^nq_1$   $\circlearrowleft$ $a^nb^na^nb^ncq_0$.

\item[$\bullet$] If $w \in L(\mathcal{A})$, then $|w|_a=|w|_b$. This is because the deletion of $a$ or $b$ will happen using the rule $(q_1,ba,q_1)$ only.
\end{itemize}
\end{example}

Now, we prove that the language $L_c$ is not in the language class $\bf{GRCOWJ}$.

\begin{lemma}\label{LcGRC}
The language $~L_c~$ of Example \ref{Lc} cannot be accepted by any $\it{GRCOWJFA}$.
\end{lemma}

\begin{proof}
Let $L_c$ be accepted by a $\it{GRCOWJFA}$ $\mathcal{A}=(\Sigma,Q,q_0,F,R)$, then $L(\mathcal{A})=L_c$ and let $l=|Q|$, $m=\{|w|~:~ (p,w,q) \in R\}$ and $n > lm$. From Example \ref{Lc}, we have $a^nb^na^nb^nc \in L_c=L(\mathcal{A})$. Then, there would be a sequence of transitions such that $q_0a^nb^na^nb^nc \circlearrowright ^* q_f$, where $q_f \in F$. The automaton can delete $a^i$, $a^ib^j$, $b^j$, $b^ja^i$, $b^jc$ or $c$, where $i,j \geq 1$, from the initial configuration $q_0a^nb^na^nb^nc$.

\vspace{0.2cm}
\noindent If it deletes $a^ib^j$, $i,j \geq 1$, from the initial configuration, then it follows the following sequence of transitions: $q_0a^nb^na^nb^nc$ $\circlearrowright$ $q_1b^{n-j}a^nb^nca^{n-i}$ $\circlearrowright^*$ $q_f$, where $q_1 \in Q$. In this case, the automaton will also accept the word $a^ib^na^nb^nca^{n-i}$ because $q_0a^ib^na^nb^nca^{n-i}$ $\circlearrowright$ $q_1b^{n-j}a^nb^nca^{n-i}$ $\circlearrowright^*$ $q_f$. Hence, $a^ib^na^nb^nca^{n-i}$ $\in L(\mathcal{A})$ but $a^ib^na^nb^nca^{n-i}$ $\not\in L_c$. Hence, it cannot delete $a^ib^j$ from the initial configuration $q_0a^nb^na^nb^nc$. Similarly, it can be shown that it cannot delete $b^j$, $b^ja^i$, $b^jc$ or $c$, where $i,j \geq 1$, from the initial configuration.

\vspace{0.2cm}
\noindent The automaton cannot delete all $a$'s starting from the initial configuration $q_0a^nb^na^nb^nc$ without deleting any $`b$' or $`c$' because in this case the automaton will have to loop and the loop has only $a$'s. Then, that loop can be used multiple times to create a mismatch in the indices of $a$ and $b$. Hence, there would be a word $w \in L(\mathcal{A})$ such that $|w|_a \neq |w|_b$ but $w \not\in L_c$. Hence, it cannot delete all $a$'s starting from the initial configuration $q_0a^nb^na^nb^nc$ without deleting a $`b$' or $`c$'.

\vspace{0.2cm}
\noindent Suppose the automaton deletes  $a^t$ for $1 \leq t < n$ without using loop, then starting from the initial configuration $q_0a^nb^na^nb^nc$, i.e., $q_0a^nb^na^nb^nc$ $~\circlearrowright^*$ $q_1a^{n-t}b^na^nb^nc$ $\circlearrowright^*$ $q_f$, where $q_1 \in Q$. Here, $q_1$ cannot delete only $a$'s and, hence, it will delete $a^ib^j$, where $0 \leq i < (n-t), 1 \leq j \leq m$, or $b^ja^i$, where $1 \leq j < m, 1 \leq i < m $, or $b^jc$, where $0 \leq j < m$. Suppose it deletes $a^ib^j$, where $0 \leq i < (n-t), 1 \leq j \leq m$, then $q_0a^nb^na^nb^nc$ $\circlearrowright^*$ $q_1a^{n-t}b^na^nb^nc$ $\circlearrowright$ $q_2b^{n-j}a^nb^nca^{n-t-i}$ $\circlearrowright^*$ $q_f$, where $q_2 \in Q$. In this case, the automaton will also accept the word $a^{t+i}b^na^nb^nca^{n-t-i}$ because $q_0a^{t+i}b^na^nb^nca^{n-t-i}$ $\circlearrowright^*$ $q_1a^ib^na^nb^nca^{n-t-i}$ $\circlearrowright$ $q_2b^{n-j}a^nb^nca^{n-t-i}$ $\circlearrowright^*$ $q_f$. Hence, $a^{t+i}b^na^nb^nca^{n-t-i}$ $\in L(\mathcal{A})$ but $a^{t+i}b^na^nb^nca^{n-t-i}$ $\not\in L_c$. Hence, $q_1$ cannot delete $a^ib^j$. Similarly, it can be shown that $q_1$ cannot delete $b^ja^i$ or $b^jc$.

\vspace{0.2cm}
\noindent Hence, we conclude $L_c \not\in$ $\bf{GRCOWJ}$.
\end{proof}

From Example \ref{Lc} and Lemma \ref{LcGRC}, we have that $\bf{GLCOWJ}$ $\not\subseteq$ $\bf{GRCOWJ}$. Similarly, it can be shown that $\bf{GRCOWJ \not\subseteq GLCOWJ}$. Hence, we have the following result.
\begin{proposition}
The language classes $\bf{GRCOWJ}$ and $\bf{GLCOWJ}$ are incomparable.
\end{proposition}

\section{$\bf{GCOWJFA}$ and $\bf{Chomsky~ hierarchy}$}\label{gcch}
This section compares the language class $\bf{GRCOWJ}$ with the language classes of Chomsky hierarchy. The class of regular languages, context free languages and context sensitive languages are denoted by $\bf{REG}$, $\bf{CF}$ and $\bf{CS}$ respectively. The relation $\bf{REG \subset ROWJ}$ was proved in \cite{owjfa} and in \cite{glowjfa}, the inclusion $\bf{ROWJ \subset GRLOWJ}$ was shown. Hence, $\bf{REG \subset ROWJ}$ $\subset \bf{GRLOWJ}$. The class $\bf{CF}$ is incomparable with both the language classes $\bf{ROWJ}$ \cite{owjfa} and $\bf{GRLOWJ}$ \cite{glowjfa}. It was shown in \cite{glowjfa} that $\bf{ROWJ \subset GRLOWJ}$ $\bf{\subset CS.}$ In this paper, we show that the language class $\bf{REG}$ is properly contained in the class $\bf{GRCOWJ}$, the language classes $\bf{CF}$ and $\bf{GRCOWJ}$ are incomparable and the class $\bf{GRCOWJ}$ is properly contained in the class $\bf{CS}$. Same results hold good for the class $\bf{GLCOWJ}$. 

We recall some results from \cite{owjfa}.
\begin{lemma}\label{regrowj}
$\bf{ROWJ}$ properly includes $\bf{REG}$, i.e., $\bf{REG \subset ROWJ}$.
\end{lemma}

\begin{lemma}\label{cfrowj}
$\bf{CF}$ and $\bf{ROWJ}$ are incomparable.
\end{lemma}
From Lemma \ref{regrowj} and Note \ref{owjgcowj}, we have the following result.
\begin{proposition}\label{reggrc}
$\bf{REG \subset GRCOWJ}$.
\end{proposition}

\begin{proposition}\label{cfgrl}
The language classes $\bf{CF}$ and $\bf{GRCOWJ}$ are incomparable.
\end{proposition}
\begin{proof}
The Dyck language $D$ $\in$ $\bf{CF}$ but from Lemma \ref{GRCD} we have, $D \not\in$ $\bf{GRCOWJ}$. Hence, $\bf{CF}$ $\not\subseteq$ $\bf{GRCOWJ}$. From Lemma \ref{cfrowj}, we conclude $\bf{ROWJ \not\subseteq CF}$ and hence, from Note \ref{owjgcowj}, we have $\bf{GRCOWJ}$ $\not\subseteq$ $\bf{CF}$.
\end{proof}

\begin{proposition}\label{csgrl}
The language class $\bf{GRCOWJ}$ is a proper subset of the class $\bf{CS}$, i.e., $\bf{GRCOWJ \subset CS}$.
\end{proposition}
\begin{proof}
Here, we give an idea on how the rules of a $\it{GRCOWJFA}$ are simulated by a linear bounded automaton $(\it{LBA})$.
\begin{enumerate}
    \item[1.] For a word $w=w_1w_2 \cdots w_n$ from the input alphabet, the tape of the $LBA$ contains $w_1w_2 \cdots w_n$, the tape head is at $w_1$ and the $LBA$ is in the starting state $q_0$, the starting state of $\it{GRCOWJFA}$. We assume, the tape is bounded from the left.

    \item[2.] Suppose the $LBA$ reaches in a state $p$ during its computation. Then, the $LBA$ searches for the nearest subword, from $\Sigma_p$, of the word available on the tape on the right side of the tape head.
    
    \item[3.] Suppose there is a nearest subword $x=x_1x_2 \cdots x_i$ and the corresponding rule of $\it{GRCOWJFA}$ is $(p,x, q)$. Then, the $LBA$ marks it as $** \cdots *$ ($i$- times). We assume $*$ is not a member of the alphabet set. By nearest subword what we mean that suppose there is a word $uxv$ from the input alphabet on the tape on the right side of the tape head, then $u$ does not contain any word from $\Sigma_p$ as a subword and $u'x' \neq x$, where $u'$ is a nonempty suffix of $u$ and $x'$ is a nonempty prefix of $x$.
    
    \item[4.] If there is a word from the alphabet left side of $*$'s, then the $\it{LBA}$ cuts that word and pastes it right side of the word available right side of $*$'s, i.e., suppose after marking $x_1x_2 \cdots x_i$ as $** \cdots *$, the tape contains $u ** \cdots * v$. Then, the $LBA$ cuts $u$ and pastes after $v$ and it goes to the state $q$ and the tape head is at the first symbol of $vu$.
    
    \item[5.] Suppose there are no symbols from the alphabet set present on the tape and the $LBA$ is in a state which is a final state of the $\it{GRCOWJFA}$, then the $LBA$ accepts the input word.
    \end{enumerate}    

Hence, we conclude that $\bf{GRCOWJ \subseteq CS}$. But from Lemma \ref{GRCD}, we have $D \not\in $ $\bf{GRCOWJ}$. Hence, $\bf{GRCOWJ \subset CS}$. Similarly, it can be shown that $\bf{GLCOWJ \subset CS}$.
   \end{proof}

\section{Closure Properties}\label{closure}
In this section we discuss closure properties of the language ~classes $\bf{GRCOWJ}$ and $\bf{GLCOWJ}$ under set theoretic as well as language operations. The classes $\bf{ROWJ}$ \cite{owjfa} and $\bf{GRLOWJ}$ \cite{glowjfa} are not closed under intersection, concatenation and reversal. Moreover, $\bf{ROWJ}$ is not closed under the operation Kleene star \cite{owjfa}, union, complement and homomorphism \cite{porowjfa}. In the same line, we show that the class $\bf{GRCOWJ}$ is not closed under intersection, concatenation, reversal and Kleene star. The same results hold true for the class $\bf{GLCOWJ}$. We begin the section with an Example of a $\it{GRCOWJFA}$. The language of the $\it{GRCOWJFA}$ will be used to prove non-closure of the class $\bf{GRCOWJ}$ under intersection, concatenation and Kleene star.

\begin{example}\label{ab}
Consider the $\it{GRCOWJFA}$ $\mathcal{A}=(\{a,b\},\{q_0,q_1\},q_0,\{q_0,q_1\},R)$, where $R=\{(q_0,ab,q_1),(q_1,ba,q_0)\}.$ Words $a^nb^n \in L(\mathcal{A})$ for all $n \geq 0$ because $q_0a^nb^n \circlearrowright q_1b^{n-1}a^{n-1} \circlearrowright q_0a^{n-2}b^{n-2} \circlearrowright \cdots \circlearrowright q_1/q_0$. If $n$ is even, then the automaton will be in the state $q_0$ and if $n$ is odd, then the automaton will be in the state $q_1$. Let $L(\mathcal{A})=L_{ab}$.
\end{example}
First we prove that the class $\bf{GRCOWJ}$ is not closed under Kleene star. We show that there exists a language $L$ in the class $\bf{GRCOWJ}$ but Kleene star $L^*$ of $L$ is not in the class $\bf{GRCOWJ}$.
\begin{proposition}
The language class $\bf{GRCOWJ}$ is not closed under Kleene star.
\end{proposition}
\begin{proof}
It can be shown, by adding some new states and rules in Example \ref{ab}, that the language $cL_{ab} \in$ $\bf{GRCOWJ}$. Let $L=(cL_{ab})^*$. We show that $L \not\in \bf{GRCOWJ}$. If $L \in \bf{GRCOWJ}$, then there exists a $\it{GRCOWJFA}$ $\mathcal{A}=(\Sigma,Q,q_0,F,R)$ such that $L(\mathcal{A})=L$. Let $l=|Q|$, $m=max\{|w|~:~ (p,w,q) \in R\}$ and $n > lm.$ Clearly, $ca^nb^nc \in L=L(\mathcal{A}).$ Hence, there exists a sequence of transitions such that $q_0ca^nb^nc \circlearrowright^* q_f$, where $q_f \in F$. The automaton can delete $c$, $ca^i$, $a^i$, $a^ib^j$, $b^j$ or $b^jc$, where $i,j \geq 1$, from the initial configuration $q_0ca^nb^nc$.

\begin{itemize}
    \item[Case 1.] If it deletes $a^i$, then the automaton makes the following sequence of transitions $q_0ca^nb^nc$ $\circlearrowright$ $q_1a^{n-i}b^ncc$ $\circlearrowright^*$ $q_f$, where $q_1 \in Q$. Then, $a^nb^ncc \in L(\mathcal{A})$ because of Lemma \ref{confi}. But $a^nb^ncc \not\in L$. Hence, the automaton cannot delete $a^i$ from the initial configuration. Similarly, it can be shown that it cannot delete $a^ib^j$, $b^j$ or $b^jc$ from the initial configuration $q_0ca^nb^nc$.

    \item[Case 2.] Hence, the automaton will have to delete $`c$' or $ca^i$ from the initial configuration $q_0ca^nb^nc$. We deal with the case $`c$' only, the case of $ca^i$ can be dealt similarly. If it deletes $c$ from the initial configuration $q_0ca^nb^nc$. Then, $q_0ca^nb^nc$ $\circlearrowright$ $q_1a^nb^nc$ $\circlearrowright^*$ $q_f$, where $q_1 \in Q$. It can delete $a^{i_1}$, $a^{i_1}b^{j_1}$, $b^{j_1}$, $b^{j_1}c$ or $c$, where $i_1,j_1 \geq 1$, from the configuration $q_1a^nb^nc$.
    
    \item[Sub-Case 2.1.]  If it deletes $a^{i_1}b^{j_1}$ from the configuration $q_1a^nb^nc$, then $q_0ca^nb^nc$ $\circlearrowright$ $q_1a^nb^nc$ $\circlearrowright$ $q_2b^{n-j_1}ca^{n-i_1}$ $\circlearrowright^*$ $q_f$, where $q_2 \in Q$. Then, $ca^{i_1}b^nca^{n-i_1} \in L(\mathcal{A})$ because of Lemma \ref{confi}. But $ca^{i_1}b^nca^{n-i_1} \not\in L$. Similarly, it cannot delete $b^{j_1}$ or $b^{j_1}c$ from the configuration $q_1a^nb^nc$.
    
    \item[Sub-Case 2.2.] If it deletes $c$ from the configuration $q_1a^nb^nc$. Then, $q_0ca^nb^nc$ $\circlearrowright$ $q_1a^nb^nc$ $\circlearrowright$ $q_2a^nb^n$ $\circlearrowright^*$ $q_f$, where $q_2 \in Q$. In this case, $b^k \not\in \Sigma_{q_1}$ for $1 \leq k \leq m$. Now, consider a word $cb^nca^n$. The word $cb^nca^n \in L(\mathcal{A})$ because $q_0cb^nca^n$ $\circlearrowright$ $q_1b^nca^n$ $\circlearrowright$ $q_2a^nb^n$ $\circlearrowright^*$ $q_f$. But $cb^nca^n \not\in L$.
    
    \item[Sub-Case 2.3.] Hence, the automaton will have to delete $a^{i_1}$ from the configuration $q_1a^nb^nc$. The automaton cannot delete all $a$'s starting from the configuration $q_1a^nb^nc$ without deleting $b$ or $c$ because in this case it will use a loop, because of the choice of $n$, consisting of only $a$'s to delete all $a$'s and the same loop can be used multiple times to create a difference in the indices of $a$ and $b$ and, hence there will be a word $w \in L(\mathcal{A})$ such that $|w|_a \neq |w|_b$ but, such $w \not\in L$. Suppose it deletes $a^t$, $1 \leq t < n$,  without using loop, starting from the configuration $q_1a^nb^nc$ without deleting $b$ or $c$, then $q_0ca^nb^nc$ $\circlearrowright$ $q_1a^nb^nc$ $\circlearrowright^*$ $q_2a^{n-t}b^nc$ $\circlearrowright^*$ $q_f$. From the configuration $q_2a^{n-t}b^nc$, it can delete $a^{i'}b^{j'}$, where $0 \leq i' <(n-t), 1 \leq j' \leq m$, or $b^{j'}c$, where $0 \leq j' < m$. In this case, it will accept $ca^{t+i'}b^nca^{n-t-i'}$ or $ca^tb^nca^{n-t}$, depending on the deletion of $a^{i'}b^{j'}$ or $b^{j'}c$ respectively. But neither $ca^{t+i'}b^nca^{n-t-i'}$ nor $ca^tb^nca^{n-t}$ are elements of $L$.
\end{itemize}
Hence, from above discussed cases, we conclude $cL_{ab} \in$ $\bf{GRCOWJ}$ but $(cL_{ab})^* \not\in \bf{GRCOWJ}$.
\end{proof}

To prove other closure properties of the class $\bf{GRCOWJ}$, we give a Lemma which gives a language that is not in the class $\bf{GRCOWJ}$. The Lemma can be proved similar to Lemma \ref{GRCD}.
\begin{lemma}\label{a^nb^n}
There does not exist any $\it{GRCOWJFA}$ that accepts the language $\{a^nb^n~|~ n \geq 0 \}$.
\end{lemma}

Now, we discuss other closure properties of the language class $\bf{GRCOWJ}$.
\begin{proposition}
$\bf{GRCOWJ}$ is not closed under
\begin{itemize}
    \item[1.] intersection;
    \item[2.] concatenation;
    \item[3.] reversal;
    \item[4.] Kleene star.
\end{itemize}
\end{proposition}
\begin{proof}
From Proposition \ref{reggrc} and Example \ref{ab}, we have $a^*b^*, L_{ab} \in \bf{GRCOWJ}$. But $a^*b^* \cap L_{ab}= \{a^nb^n~|~ n \geq 0\} \not\in \bf{GRCOWJ}$, by Lemma \ref{a^nb^n}. Hence, $\bf{GRCOWJ}$ is not closed under intersection.

From Example \ref{ab} and Proposition \ref{reggrc}, we have $ L_{ab},\{c\} \in \bf{GRCOWJ}$. Similar to Lemma \ref{LcGRC}, we can prove $L_{ab}\{c\} \not\in \bf{GRCOWJ}$. Hence, $\bf{GRCOWJ}$ is not closed under concatenation.

It was shown in \cite{owjfa} that the languages $c\{w \in \{a,b\}^*~|~|w|_a=|w|_b\} \in \bf{ROWJ}$. From Note \ref{owjgcowj}, we have $\bf{ROWJ \subset GRCOWJ}$ and hence $c\{w \in \{a,b\}^*~|~|w|_a=|w|_b\} \in \bf{GRCOWJ}$. Similar to the proof given in Lemma \ref{LcGRC}, one can show that $\{w \in \{a,b\}^*~|~|w|_a=|w|_b\}c \not\in \bf{GRCOWJ}$,  Hence, $\bf{GRCOWJ}$ is not closed under reversal.
\end{proof}

\section{Conclusion}\label{con}
In this paper, we have introduced a new extended version of the model one-way jumping finite automata ($\it{OWJFA}$) called generalized circular one-way jumping finite automata ($\it{GCOWJFA}$). By extended version, we mean that $\it{OWJFA}$ change its states by deleting a letter of an input word whereas $\it{GCOWJFA}$ change its states by deleting a subword of an input word. The newly introduced version $\it{GCOWJFA}$ is different from one of the extended versions generalized linear one-way jumping finite automata ($\it{GLOWJFA}$) of $\it{OWJFA}$. Like $\it{OWJFA}$ and $\it{GLOWJFA}$, $\it{GCOWJFA}$ also has two variants called generalized right circular one-way jumping finite automata ($\it{GRCOJFA}$) and generalized left circular one-way jumping finite automata ($\it{GLCOWJFA}$). We have shown that the newly defined version $\it{GCOWJFA}$ is more powerful than $\it{OWJFA}$. We also show that its computational power is distinct from that of $\it{GLOWJFA}$. We have compared the computational power of the variants $\it{GRCOWJFA}$ and $\it{GLCOWJFA}$. The language classes of $\it{GRCOWJFA}$ and $\it{GLCOWJFA}$ are compared with the language classes of Chomsky hierarchy. Finally, closure properties of the language classes of $\it{GRCOWJFA}$ and $\it{GLCOWJFA}$ have been investigated.

\end{document}